\documentclass[10pt,conference,letter]{IEEEtran}
\usepackage{bm}                                                                                                         
\usepackage{amssymb}
\usepackage{amssymb}
\usepackage{amsxtra}
\usepackage{amsmath}
\usepackage{latexsym}
\usepackage{verbatim}
\usepackage{setspace}
\usepackage{theorem}
\usepackage{rotating}
\usepackage[dvips]{color}
\usepackage{mediabb}

\renewcommand{\baselinestretch}{.92}
\theorembodyfont{\normalfont\slshape}

\newtheorem{theorem}{Theorem}

\newtheorem{definition}{Definition}

\newtheorem{remark}{Remark}

\def\S{\Sigma}

\def\bb{{\bm b}}
\def\bf{{\bm f}}
\def\bs{{\bm s}}
\def\bu{{\bm u}}
\def\bv{{\bm v}}

\def\bx{{\bm x}}
\def\by{{\bm y}}

\def\cA{{\mathcal A}}

\def\cF{{\mathcal F}}

\def\cS{{\mathcal S}}

\def\rara{\textcolor{red}{\xrightarrow{\textcolor{black}{\bm a}}}}
\def\rbra{\textcolor{blue}{\xrightarrow{\textcolor{black}{\bm a}}}}
\def\rar{\textcolor{red}{\rightarrow}}
\def\bar{\textcolor{blue}{\rightarrow}}

\def\figpath{.}

\def\REP#1#2{{\color{black}#2}}
\def\INS#1{{\color{black}#1}}
\def\DEL#1{}
\def\REPa#1#2{{\color{black}#2}}
\def\INSa#1{{\color{black}#1}}

\begin{document}

\thispagestyle{empty}
\title{A Graph Representation for Two-Dimensional\\
Finite Type Constrained Systems}

\author{
\authorblockN{Takahiro Ota}
\authorblockA{Dept.\ of Computer \& Systems Engineering\\
Nagano Prefectural Institute of Technology\\
813-8, Shimonogo, Ueda, Nagano, 386-1211, JAPAN\\
Email: \texttt{ota@pit-nagano.ac.jp}}
\and 
\authorblockN{Akiko Manada and Hiroyoshi Morita}
\authorblockA{Graduate School of Information Systems\\
The University of Electro-Communications\\
1-5-1, Chofugaoka, Chofu, Tokyo, 182-8585, JAPAN\\
Email: \texttt{\{amanada, morita\}@is.uec.ac.jp}}
}

\maketitle


\renewcommand{\thefootnote}{\arabic{footnote}}
\setcounter{footnote}{0}

\begin{abstract}
The demand of two-dimensional source coding and constrained coding 
has been getting higher these days, 
but compared to the one-dimensional case, 
many problems have remained open as the analysis is cumbersome. 
A main reason for that would be because there are 
no graph representations discovered so far. 
In this paper, we focus on  a two-dimensional finite type constrained system, a set of two-dimensional 
blocks characterized by a finite number of two-dimensional constraints, and 
propose its graph representation. 
We then show how to generate an element of the two-dimensional finite type constrained system 
from the graph representation.
\end{abstract}

\section{Introduction\label{intro}}
In one-dimensional source coding, a graph representation (\emph{e.g}, labelled directed graph
for a given string) or a finite-state source (\emph{e.g.}, Markov source) is utilized as a probabilistic model.
The graph representation is useful not only for data compression in practical but also
for analysis of the data compression~\cite{MT95}. 
Moreover, in one-dimensional channel coding,
constrained coding is utilized for reducing the likelihood of 
errors by removing data sequences that can be easily affected 
by the predictable noise. 
The study of one-dimensional constrained coding is based on 
the study of one-dimensional constrained systems, which can be 
represented by labelled directed graphs. Indeed, many 
important results are derived by analyzing the 
characteristics of those graphs~\cite{LM95}. 

On the other hand, 
the study of two or higher dimensional cases
is cumbersome, and many important problems (\emph{e.g.}, \INS{probabilistic models,} the capacity, the existence of 
infinite arrays not containing forbidden patterns) are still open or known to be 
unsolvable in finite steps in general \REP{\cite{Lin04}}{\cite{Lin04, JF10}}. 
The main reason for this would be because  
there are no explicit way to present such high-dimensional source and constrained coding using finite graphs~\cite{JF10}.
Conversely, if some graph representation is proposed, then such a representation can be a 
breakthrough to approach open problems. 

In this paper, 
for a given finite set $\cF$ of two-dimensional forbidden blocks~(forbidden set), 
we focus on the two-dimensional Finite Type Constrained System (2D-FTCS) which is a 
set of blocks that do not contain any forbidden block in $\cF$ as subblock.
We construct a labelled directed graph based on forbidden blocks and then
show how to  generate an arbitrary block in the 2D-FTCS from the graph. 
Thus, the existence of such a graph presentation 
can be used to answer the existence of blocks in the 2D-FTCS, 
and preferably, to explicitly compute its capacity.

The rest of the paper is organized as follows. 
We first give in Section~\ref{Notations} basic notations and definitions.
In Section~\ref{Proposed}, we propose a labelled graph representation for a 2D-FTCS and
prove that the graph representation generates a block if and only if
the block is an element of the 2D-FTCS. We then show in Section~\ref{howto} a process to generate 
an arbitrary block in the 2D-FTCS from the graph. 
We terminate this paper with conclusion and future works in Section~\ref{conclusion}.

\section{Basic Notations and Definitions}\label{Notations}
\subsection{Alphabet and Block}
Let $\S$ be an \emph{alphabet}, a finite set of symbols. 
We denote by  $\S^{(m,n)}$ the set of $m\times n$ finite blocks $\bb=(b_{i,j})_{1\le i\le m, 1 \le j\le n}$ over $\S$, where $b_{i,j} \in \S$ is the element of $\bb$ at $(i,j)$-coordinate. For simplicity, we always assume that blocks are finite. Furthermore, define 
$\S^{(*,*)}=\cup_{m,n\ge 0}\S^{(m,n)}, $
where $\S^{(m,n)}$ consists only of the \emph{empty block} ${\bm \lambda}$ 
when at least one of $m$ and $n$ is $0$.  
Given a block $\bb \in \S^{(*,*)}$, denote by $|\bb|_r$ and $|\bb|_c$  the length of row (the \emph{height}) and the length of column (the \emph{width}), respectively.   
For example, when $\S = \{0, 1\}$, Fig.~\ref{fig:ex33s} illustrates $\bb \in \S^{(3, 3)}$ where $|\bb|_r\!=\!|\bb|_c\!=\!3$.

\begin{figure}[htb]
	\begin{minipage}{0.2\hsize}
		\begin{center}
			\includegraphics[clip, width=0.8\linewidth]{\figpath/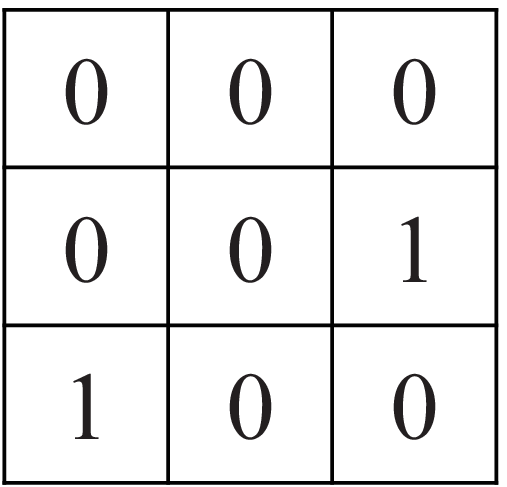}
		\end{center}
		\caption{A $3 \times 3$ block $\bb$.}
		\label{fig:ex33s}
	\end{minipage}
	\begin{minipage}{0.75\hsize}
		\begin{center}
			\includegraphics[clip, width=0.8\linewidth]{\figpath/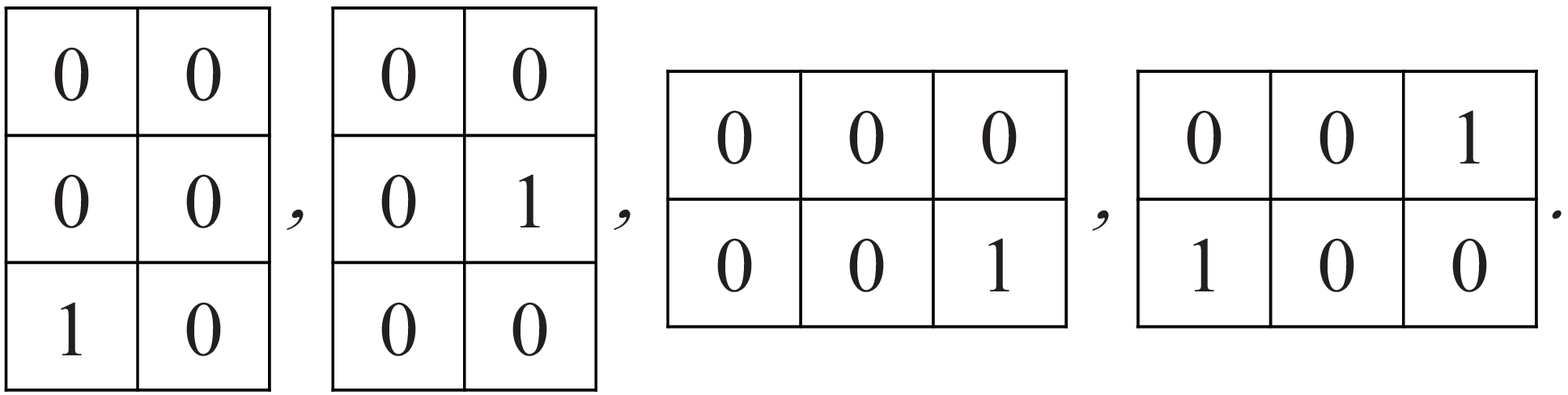}
		\end{center}
		\caption{$\pi_c(\bb)$, $\sigma_c(\bb)$, $\pi_r(\bb)$, and $\sigma_r(\bb)$ of $\bb$ in Fig.~\ref{fig:ex33s}.}
		\label{fig:ps2}
	\end{minipage}
\end{figure}

\subsection{Subblock, Prefix, Suffix and Concatenation}
Throughout this subsection, we always let $\bb=(b_{i,j}) \in \S^{(m,n)}$ and $\bb'=(b'_{i',j'})\in \S^{(m',n')}$ be
two blocks of size $m\times n$ and $m' \times n'$, respectively. 

Block $\bb'$ is called a \emph{subblock} of $\bb$ when $\bb'$ appears within $\bb$; that is, 
$b'_{i',j'}=b_{k+i',l+j'}$ for some fixed non-negative integers $k,l$ and any integers 
$1\le i'\le m', 1\le j'\le n'$. 
In particular, subblock $\bb'$ is called a \emph{prefix} (\emph{resp.}, \emph{suffix}) of $\bb$
when $k=l=0$ (\emph{resp.} when 
$k=m-m'$ and $l=n-n'$). 
Among prefixes and suffixes, we often focus on the prefix and the suffix of size $m\times (n-1)$, 
and the prefix and the suffix of size $(m-1)\times n$ which we denote by $\pi_c(\bb), \sigma_c(\bb), \pi_r(\bb), \sigma_r(\bb)$, respectively. For example, for $\bb$ in Fig.~\ref{fig:ex33s}, 
Fig.~\ref{fig:ps2} shows $\pi_c(\bb)$, $\sigma_c(\bb)$, $\pi_r(\bb)$, and $\sigma_r(\bb)$
from the left-hand side.

For blocks $\bb$ and $\bb'$, when 
$m=m'$, we can consider the block $[\bb  , \bb']_c\in \S^{(m, n+ n')}$ generated by column-wisely concatenating $\bb$ and $ \bb'$.  Similarly, when $n= n'$, we can consider the block $[\bb,  \bb']_r\in \S^{(m+ m', n)}$ generated by row-wisely concatenating $\bb$ and $\bb'$.

\subsection{Two-Dimensional Finite Type Constrained Systems}
Let $h, w$ be some fixed non-negative integers. 
Given a finite set $\cF\subset \S^{(h,w)}$,
we define a \emph{Two-Dimensional Finite Type Constrained System (2D-FTCS)} $\cS_{\cF}$ to be a subset of $\S^{(*,*)}$ which can be characterized by $\cF$. More precisely, for a 2D-FTCS $\cS_{\cF}$,
a block $\bb$ is an element of $\cS_{\cF}$ if and only if $\bb$ does not contain any block $\bf \in  \cF$ as subblock. \footnote{The definition of 2D-FTCSs can vary, depending on authors. Our definition covers more general case compared with the definition in \cite{LPCM08}.}  
We call $\cF$ a \emph{forbidden set} and an element $\bf \in \cF$ a \emph{forbidden block}. 

We note that for one-dimensional case, the definition of an 1D-FTCS in general starts with the  
definition of a constrained system for which the existence of 
a graph representation (called \emph{presentation}) matters. We do not refer it at this moment, but we will discuss it in latter sections which shows that out definition is a natural extension 
of the definition of a 1D-FTCS.

When a forbidden set $\cF\subset \S^{(h,w)}$ is given, we define the \emph{allowed set} $\cA_{\cF} \subset \S^{(h,w)}$ for $\cF$ as 
$\cA_{\cF}:=\S^{(h,w)} \setminus \cF$. 
We further define the \emph{modified 2D-FTCS}
$\cS_{\cF}^{(*,*)}$ from a 2D-FTCS $\cS_{\cF}$ so that
\begin{align}
\cS_{\cF}^{(*,*)} :=\cS_{\cF}\cap \left(\bigcup_{m\ge h, n\ge w}\S^{(m,n)}\right)
\end{align}
For example, consider the 2D-FTCS $\cS_{\cF}$ characterized 
by  $\cF$ in Fig.~\ref{fig:Ft}, which is well-known 
as the hard-square constraint~\cite{JF10}, forbidding the appearance of
 $[1,1]_r$ and $[1,1]_c$. 
Fig.~\ref{fig:A} is  
the allowed set $\cA_{\cF}$ for $\cF$.
Observe that block $\bb$
in Fig.~\ref{fig:ex33s} is an element of $\cS_{\cF}$.

\begin{figure}[htb]
	\begin{minipage}{0.45\hsize}
		\begin{center}
			\includegraphics[clip, width=1.0\linewidth]{\figpath/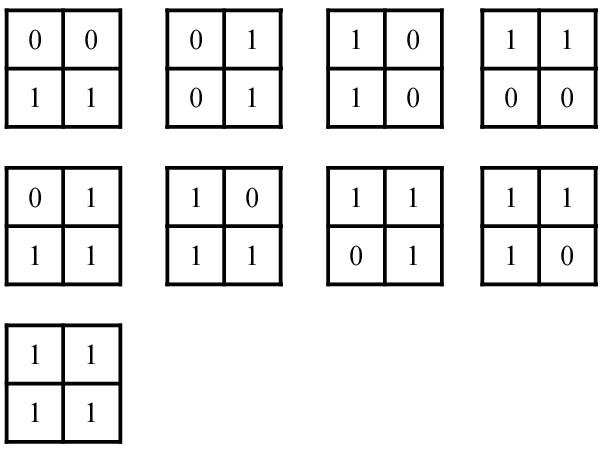}
		\end{center}
		\caption{Forbidden blocks in $\cF$.}
		\label{fig:Ft}
	\end{minipage}
	\begin{minipage}{0.55\hsize}
		\begin{center}
			\includegraphics[clip, width=0.81\linewidth]{\figpath/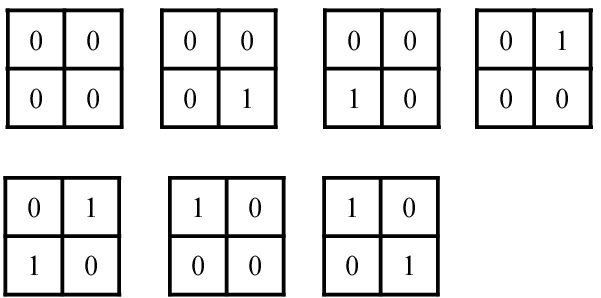}
		\end{center}
		\caption{Blocks in the allowed set $\cA_{\cF}$ for $\cF$.}
		\label{fig:A}
	\end{minipage}
\end{figure}

\section{A Graph Representation for 2D-FTCSs}~\label{Proposed}
The main contribution of this paper is to propose 
a graph representation for 2D-FTCS $\cS_{\cF}$. 
We only focus on graph representations for $\cS_{\cF}^{(*,*)}$,
but observe that it is enough since $\cS_{\cF}\setminus \cS_{\cF}^{(*,*)}$ 
is a finite set.

\subsection{Notations and Definitions}
Let $G=(V,E, \ell)$ be a labelled directed graph consisting of 
vertex set $V$, labelled edge set $E \subset V\times V$ and 
an edge labelling $\ell: E\rightarrow \S^{(*,*)}$.
  
Now suppose a 2D-FTCS $\cS=\cS_{\cF}$ characterized by finite forbidden set $\cF\subset \S^{(h,w)}$
is given. We construct labelled directed graphs $G_r(\cS)$ and $G_c(\cS)$
from $\cS$ as follows. 

\begin{definition}
Given a 2D-FTCS $\cS=\cS_{\cF}$ characterized by $\cF\subset \S^{(h,w)}$,  
the row-wise presentation $G_r(\cS)=(V_r, E_r, \ell_r)$ of $\cS$ is a labelled directed graph satisfying 
\begin{itemize}
\item $V_r$ is the allowed set $\cA_{\cF}$ for $\cF$; and
\item draw a blue (row-wise) edge 
from $\bu$ to $\bv$  labelled \REP{$\bm a(=\ell_r(\bu, \bv)) \in \S^{(h,1)}$}{{$\bm a(=\ell_r(\bu, \bv)) \in \S^{(1,w)}$}} 
(\it{i.e.}, $\bu \rbra \bv$) if and only if $\sigma_c(\bm u)=\pi_c(\bm v)$ 
and the $h$-th row of $\bm v$ is $\bm a$.
\end{itemize}
Similarly, the column-wise presentation $G_c(\cS)=(V_c,E_c, \ell_c)$ of $\cS$ is a labelled directed graph satisfying  
\begin{itemize}
\item $V_c$ is the allowed set $\cA_{\cF}$ for $\cF$; and
\item draw a red (column-wise) edge
from $\bu$ to $\bv$  labelled \REP{$\bm a(=\ell_c(\bu, \bv)) \in \S^{(1,w)}$}{$\bm a(=\ell_c(\bu, \bv)) \in \S^{(h,1)}$}
(\it{i.e.}, $\bu \rara \bv$) if and only if $\sigma_r(\bm u)=\pi_r(\bm v)$ 
and the $w$-th column of $\bv$ is $\bm a$.
\end{itemize}
\end{definition}

For simplicity, we omit edge labels. Furthermore, we identify a vertex $\bv \in \cA_{\cF}$ with its \emph{identifier} $g(\bv)$; a positive integer defined
via  bijection $g: \cA_{\cF} \rightarrow \{1, 2, \dots, ||\cA_{\cF}||\}$, where $||\cA_{\cF}||$ is the cardinality of $\cA_{\cF}$.
For example, Fig.~\ref{fig:Gr} and Fig.~\ref{fig:Gc} show  $G_r(\cS)$ and $G_c(\cS)$ for $\cA_{\cF}$ in Fig.~\ref{fig:A}, respectively, and the number in a circle represents its identifier.

	\begin{figure}[htb]
		\begin{center}
			\includegraphics[clip, width=0.6\linewidth]{\figpath/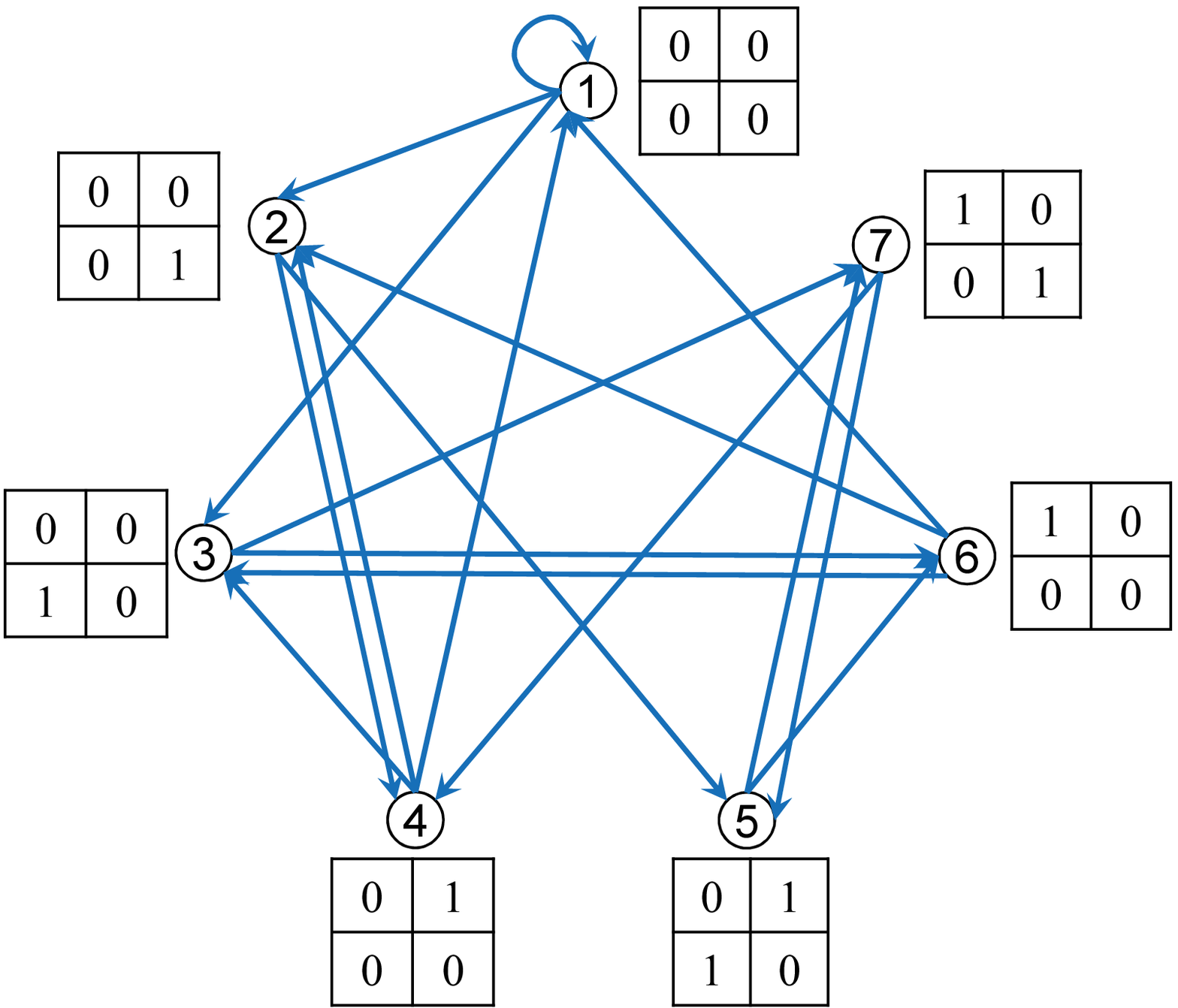}
			\caption{The row-wise presentation $G_r(\cS)$.}
			\label{fig:Gr}
		\end{center}
	\end{figure}
	\begin{figure}[htb]
		\begin{center}
			\includegraphics[clip, width=0.6\linewidth]{\figpath/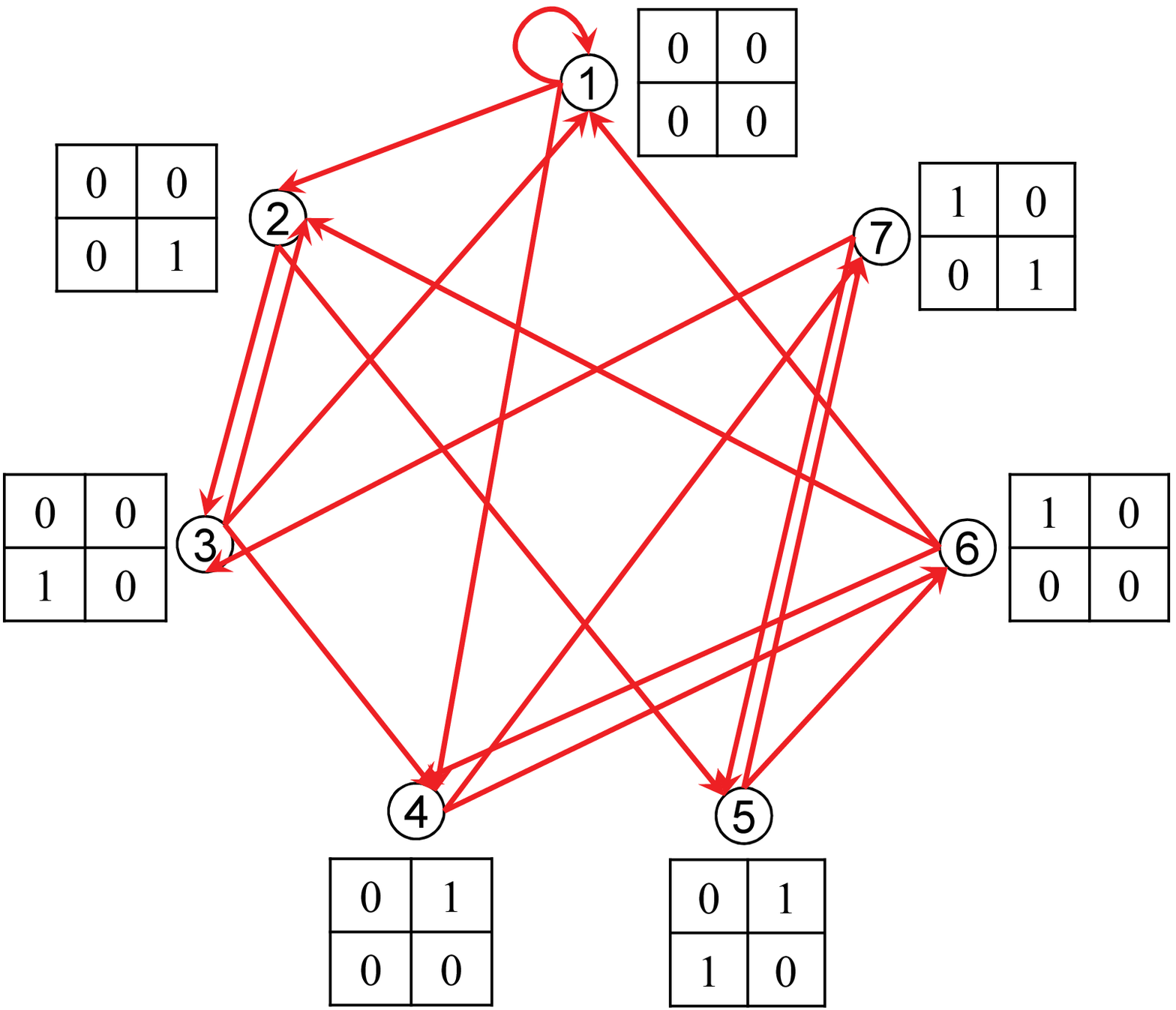}
			\caption{The column-wise presentation $G_c(\cS)$.}
			\label{fig:Gc}
		\end{center}
	\end{figure}

Now take an arbitrary block $\bx$ of height 
\REP{$|\bx|_r> h$}
{$|\bx|_r(=h') > h$}
and width $w$. 
We say a path 
\REP{$\eta : \bu_0 \bar \bu_1 \bar \cdots \bar \bu_{|\bx|_r-k}$}
{$\eta : \bu_0 \bar \bu_1 \bar \cdots \bar \bu_{h'-h}$}  
in $G_r(\cS)=(V_r, E_r, \ell_r)$  \emph{generates} the block $\bx$ 
if concatenating $\ell_{r}(\bu_i,\bu_{i+1})$ in the order of 
\REP{$i=0,1, \ldots, |\bx|_r-h-1$}
{$i=0,1, \ldots, h'-h-1$}
after $\bu_0$ 
 \INS{(the \emph{head block})} 
in row-wise generates $\bx$.
In other words, path $\eta$ generates block $\bx$ if
\REP{$$\bx=[\bu_0, \ell_{r}(\bu_0,\bu_{1}), \ell_{c}(\bu_1,\bu_{2}), \ldots, \ell_{r}(\bu_{|\bx|_r-h-1},\bu_{|\bx|_r-h})]_r. $$}
{$$\bx=[\bu_0, \ell_{r}(\bu_0,\bu_{1}), \ell_{c}(\bu_1,\bu_{2}), \ldots, \ell_{r}(\bu_{h'-h-1},\bu_{h'-h})]_r. $$}
If such a path $\eta$ exists, we say $G_r(\cS)$ \emph{generates} block $\bx$.
By convention, $G_r(\cS)$ generates any $h\times w$ block $\bx \in \cA_{\cF}$, 
assuming path $\eta$ in $G_r(\cS)$ consisting only of $\bx$.

Similarly, for a block $\by$ of height $h$ and width
\REP{$|\by|_c> w$,}
{$|\by|_c(=w') > w$,}
we say a path 
\REP{$\tau: \bv_0 \rar \bv_1 \rar \cdots \rar \bv_{|\by|_c-w}$}
{$\tau: \bv_0 \rar \bv_1 \rar \cdots \rar \bv_{w'-w}$}
in $G_c(\cS)=(V_c, E_c, \ell_c)$
\emph{generates} the block $\by$
if concatenating $\ell_{c}(\bv_i,\bv_{j+1})$ in the order of
\REP{$j=0,1, \ldots, |\by|_c-w-1$}
{$j=0,1, \ldots, w'-w-1$}
after 
\INS{the head block} $\bv_0$ 
in column-wise generates $\by$.
In other words, path $\tau$  generates block $\by$ if
\REP{$$\by=[\bv_0, \ell_{c}(\bv_0,\bv_{1}), \ell_{c}(\bv_1,\bv_{2}), \ldots, \ell_{c}(\bv_{|\by|_c-w-1},\bv_{|\by|_c-w})]_c. $$}
{$$\by=[\bv_0, \ell_{c}(\bv_0,\bv_{1}), \ell_{c}(\bv_1,\bv_{2}), \ldots, \ell_{c}(\bv_{w'-w-1},\bv_{w'-w})]_c. $$}
If such a path $\tau$ exists, we say
$G_c(\cS)$ \emph{generates} block $\by$.  
By convention, $G_c(\cS)$ generates any $h\times w$ block  $\by \in \cA_{\cF}$, 
assuming path $\tau$ in  $G_c(\cS)$ consisting only of $\by$.

\begin{remark}
From the definitions of $G_r(\cS)$ and $G_c(\cS)$,  it is straightforward to check that
$G_r(\cS)$ generates a block $\bx$ of width $w$ if and only if
each $h\times w$ subblock of $\bx$ is in $\cA_{\cF}$, which is equivalent to say 
$\bx \in \cS^{(*,w)}_{\cF}:=\cS_{\cF}\cap \left(\cup_{m\ge h}\S^{(m,w)}\right)$.
Similarly,  $G_c(\cS)$ generates a block $\by$ of height $h$ if and only if 
each $h\times w$ subblock of $\by$ is in $\cA_{\cF}$, which is equivalent to say 
$\by \in \cS^{(h,*)}_{\cF}:=\cS_{\cF}\cap \left(\cup_{n\ge w}\S^{(h,n)}\right)$.
\label{st_rem}
\end{remark}

We can naturally extend the notion of generating a certain size of block to the notion of generating any block $\bb$ of height $|\bb|_r\ge h$ and width $|\bb|_r \ge w$ as follows.  

\begin{definition}\label{def:gGrGc}
Given a block $\bb$ of height $|\bb|_r\ge h$ and width $|\bb|_r \ge w$,
we say $G_r(\cS)$ generates $\bb$
if $G_r(\cS)$ can generate
any $|\bb|_r\times w$ subblock of $\bb$, and \INSa{$G_c(\cS)$}  generates $\bb$ if $G_c(\cS)$ can generate
any $h \times |\bb|_c$ subblock of $\bb$.
\end{definition}

From Definition~\ref{def:gGrGc}, if $G_r(\cS)$ \REP{or}{and} $G_c(\cS)$ generates $\bb$, then any $h \times w$ subblock in $\bb$ is an element of the allowed set, and hence, a vertex in $G_r(\cS)$ and $G_c(\cS)$.

\subsection{Relationship between $G_r(\cS)$ and $G_c(\cS)$}

Now take any block $\bb$ such that $G_r(\cS)$ \REP{or}{and} $G_c(\cS)$ can generate. 
For a $h \times w$ subblock $\bv$ in $\bb$ whose right-bottom coordinate is
$(i, j)~(h\leq i \leq |\bb|_r, w\leq j \leq |\bb|_c$),
let $s(i, j)$ be the identifier $g(\bv)$ of the vertex $\bv$.
Since any such subblock $\bv$ is a vertex in $G_r(\cS)$ and $G_c(\cS)$,
for given $s(i-1, j-1)$ and $s(i, j)$, where $k <i \leq |\bb|_r, l < j \leq |\bb|_c$,
there must exist two paths from $s(i-1, j-1)$ to $s(i, j)$
\begin{enumerate}
\item $s(i-1, j-1) \rar s(i-1, j) \bar  s(i, j)$ 
\item $s(i-1, j-1) \bar s(i, j-1)  \rar s(i, j)$
\end{enumerate}
as shown in Fig.~\ref{fig:condT}. For example, consider $\bb$ in Fig.~\ref{fig:ex33s}, and $G_r(\cS)$ and $G_c(\cS)$ in Figs.~\ref{fig:Gr} and \ref {fig:Gc}.
For $s(2, 2)=1$ and $s(3, 3)=4$, there are two paths; $1  \rar 2(=s(2, 3)) \bar  4$ and
$1  \bar 3(=s(3, 2))  \rar 4$.

	\begin{figure}[htb]
		\begin{center}
			\includegraphics[clip, width=0.4\linewidth]{\figpath/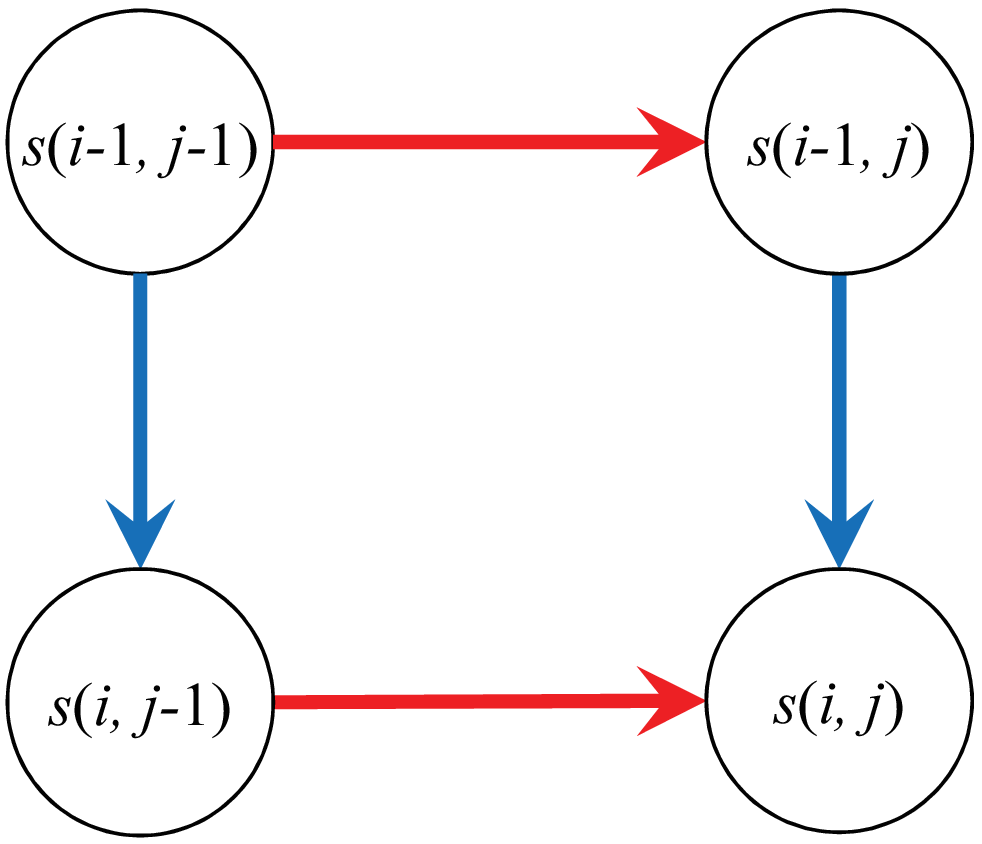}
			\caption{Paths among four adjacent vertices in $G_r(\cS)$ and $G_c(\cS)$.}
			\label{fig:condT}
		\end{center}
	\end{figure}

For $G_r(\cS)$ and $G_c(\cS)$ in Figs.~\ref{fig:Gr} and \ref{fig:Gc},  
Fig.~\ref{fig:exCond} shows all the combinations of paths among four adjacent vertices in $G_r(\cS)$ and $G_c(\cS)$. 
Four numbers at left-top, right-top, left-bottom and right-bottom in a rectangle
are represented by $s(i-1, j-1)$, $s(i-1, j)$, $s(i, j-1)$, and $s(i, j)$, respectively.
For example, the first rectangle in the third row \REP{of}{in} Fig.~\ref{fig:exCond}, 
$s(i-1, j-1)$, $s(i-1, j)$, $s(i, j-1)$ and $s(i, j)$ are given by 2, 5, 4, and 6, respectively.
In this example, there are total 63 combinations of paths among four adjacent vertices in $G_r(\cS)$ and $G_c(\cS)$.
	\begin{figure}[htbp]
		\begin{center}
			\includegraphics[clip, width=0.75\linewidth, height=0.70\linewidth]{\figpath/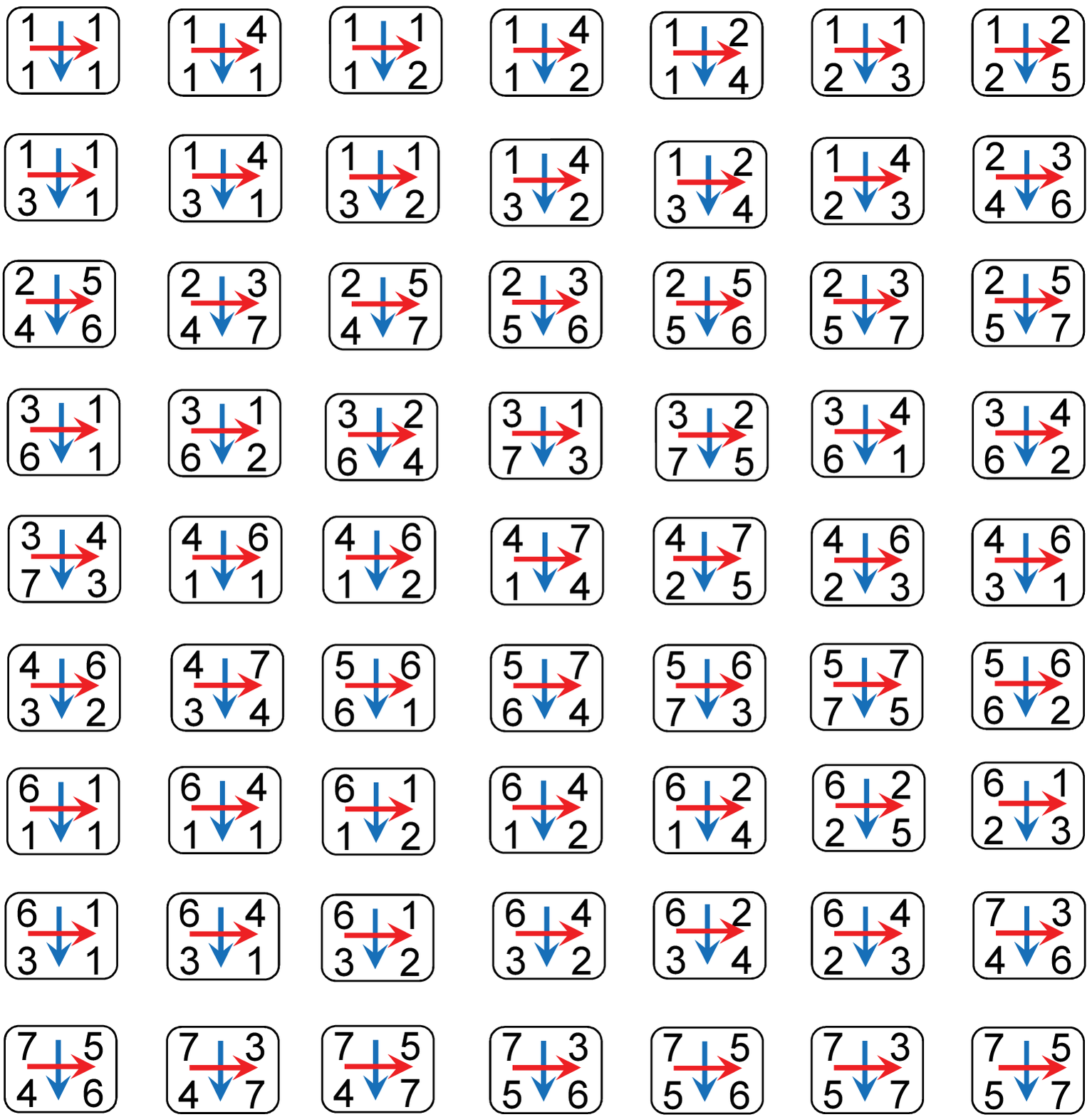}
			\caption{An example of all the combinations of paths among four adjacent vertices in $G_r(\cS)$ and $G_c(\cS)$.}
			\label{fig:exCond}
		\end{center}
	\end{figure}

Fig.~\ref{fig:exRelC} illustrates two combinations of paths among four adjacent vertices, $(1, 2, 1, 4)$ and $(2, 5, 4, 6)$, 
in $G_c(\cS)$ where $(a, b, c, d)$ represents $a = s(i-1, j-1), b = s(i-1, j), c = s(i, j-1),$ and $d = s(i, j)$.
Light green and orange rectangles represent (1, 2, 1, 4) and (2, 5, 4, 6), respectively.

	\begin{figure}[htb]
		\begin{center}
			\includegraphics[width=0.5\linewidth]{\figpath/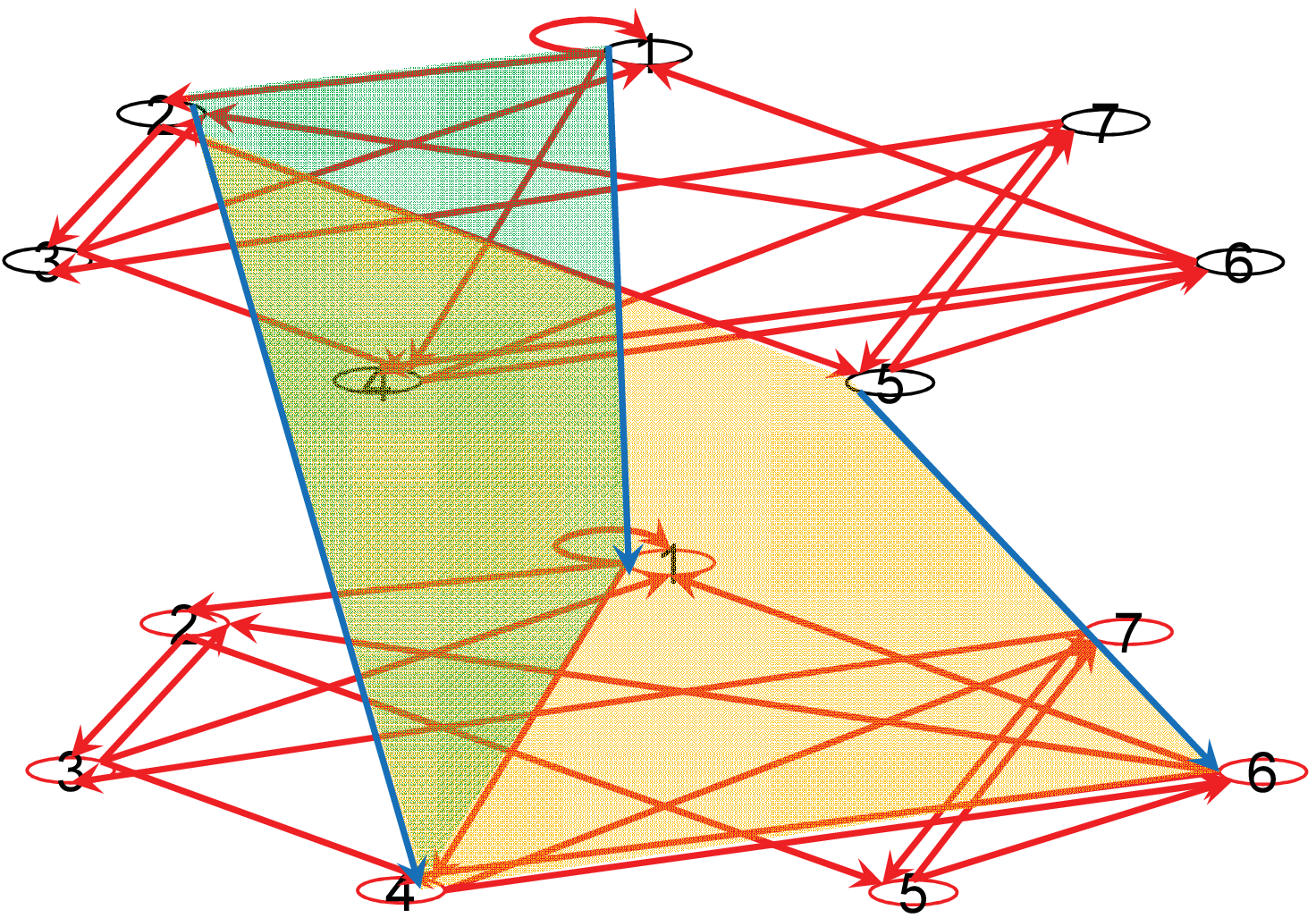}
			\caption{Two combinations of paths among four adjacent vertices, $(1, 2, 1, 4)$ and $(2, 5, 4, 6)$, in $G_c(\cS)$.}
			\label{fig:exRelC}
		\end{center}
	\end{figure}

\subsection{Graph Representation and its Paths}\label{sec:GR}
We are now in a position of combining the row-wise presentation $G_r(\cS)$ and the 
column-wise presentation $G_c(\cS)$ to obtain a
graph representation $G(\cS)$.

\begin{definition} 
We define a graph representation $G(\cS)=(V_{\cS}, E_{\cS}, \ell_{\cS})$ of a 2D-FTCS $\cS=\cS_{\cF}$ so that
\begin{itemize}
\item $V_{\cS}=V_{r}=V_{c}=\cA_{\cF}$.
\item $E_{\cS}=E_{r}\cup E_{c}$.
\item \[
\ell_{\cS}(\bu, \bv)=\begin{cases}
\ell_r(\bu, \bv) \ (\mbox{if $(\bu,\bv)\in E_r$})\\
\ell_c(\bu, \bv) \ (\mbox{if $(\bu,\bv)\in E_c$})
\end{cases}
\]
\end{itemize}
\end{definition}

Thus, $G(\cS)$ is easily obtained from $G_r(\cS)$ and $G_c(\cS)$ by combining them, so 
$G_r(\cS)$ and $G_c(\cS)$ are subgraphs of $G(\cS)$.  

We define that $G(\cS)$ generates $\bb$ as follows.

\begin{definition}
Let $\cS=\cS_{\cF}$ be a 2D-FTCS characterized by $\cF\subset \S^{(h,w)}$. 
We say graph representation $G(\cS)$ generates block $\bb$ if 
\begin{itemize}
\item any $|\bb|_r \times w$ subblock $\bb'$ of $\bb$ is generated 
by a path in $G(\cS)$ consisting only of blue edges (and hence, $\bb'$  is generated 
by a subgraph $G_r(\cS)$); and 
\item any $h \times |\bb|_c$ subblock $\hat \bb$ of $\bb$ is generated 
by  a path in $G(\cS)$ consisting only of red edges (and hence, $\hat \bb$  is generated 
by a subgraph $G_c(\cS)$).
\end{itemize}
\end{definition}

 \begin{remark}
For paths used to generate $\bb$, 
it is important to observe that they 
maintain the following conditions in terms of $s(i,j)$.

\begin{description}
\item[Case 1:]\ \ $i = h\ \text{or}\ j = w.$
\item[Case 2:]\ \ $i > h\ \text{and}\ j > w.$
\end{description}

\underline{Case 1:}\ 
When $i = h$, a prefix of $\bb$ of size $h \times |\bb|_c$ is
generated by $G_c(\cS)$. In this case, since $i\!-\!1 < h $, $s(i\!-\!1, j\!-\!1)$ does not exist. 
Hence, a path in $G_c(\cS)$ is not restricted by paths among four adjacent vertices in $G_r(\cS)$ and $G_c(\cS)$ shown in Fig.~\ref{fig:condT}. 
Moreover, for each $w\leq j \leq |\bb|_c$, $s(h, j)$ 
turns out to be the head block 
of the path $\tau$ in $G_r(\cS)$ 
generating a $|\bb|_r \times w$ subblock of $\bb$ whose 
right-bottom coordinate  is $(|\bb|_r, j)$.

Similarly, 
when $j = w$, a prefix of $\bb$ of size 
\REPa{$|\bb|_c \times w$}
{$|\bb|_r \times w$}
is
generated by $G_r(\cS)$. In this case, since $j\!-\!1 < w$, $s(i\!-\!1, j\!-\!1)$ does not exist. 
Hence, a path in $G_r(\cS)$ is not restricted by paths among four adjacent vertices in $G_r(\cS)$ and $G_c(\cS)$ shown in Fig.~\ref{fig:condT}. 
Moreover, for each $h\leq i \leq |\bb|_r$, $s(i,w)$ turns out to be the head block 
of the path $\eta$ in $G_c(\cS)$ 
generating a $h\times |\bb|_c$ subblock of $\bb$ 
 whose right-bottom coordinate is $(i, |\bb|_c)$.

\underline{Case 2:}\ 
From paths among four adjacent vertices in $G_r(\cS)$ and $G_c(\cS)$ shown in Fig.~\ref{fig:condT},
$s(i\!-\!1, j\!-\!1)$, $s(i\!-\!1, j)$, and $s(i, j\!-\!1)$ are required to determine a path to $s(i, j)$.
Moreover, there must exist two paths from $s(i\!-\!1, j\!-\!1)$ to $s(i, j)$, those are
$s(i\!-\!1, j\!-\!1) \rar s(i\!-\!1, j) \bar s(i, j)$ and $s(i\!-\!1, j\!-\!1)\bar s(i, j\!-\!1) \rar s(i, j)$.
\label{case:rem}
\end{remark}

We prove Theorem~\ref{th:GA} which is a main result of the paper.

\begin{theorem}\label{th:GA}
$G(\cS)$ generates $\bb$ if and only if $\bb \in \cS_{\cF}^{(*,*)}$.
\end{theorem}
\begin{proof}
We first prove that if $G(\cS)$ generates $\bb$ then $\bb \in \cS_{\cF}^{(*,*)}$.
If $G(\cS)$ generates $\bb$, then any $h\times w$ subblock of $\bb$ 
is a vertex in $G(\cS)$, and therefore,  an element of $\cA_{\cF}$.
Thus, $\bb$ does not contain any forbidden block as a subblock, and hence, 
 $\bb \in \cS_{\cF}^{(*,*)}$. 

We next prove that if $\bb \in \cS_{\cF}^{(*,*)}$ then $G(\cS)$ generates $\bb$.
We assume by contradiction that $G(\cS)$ cannot generate $\bu \in \cS_{\cF}^{(*,*)}$. 
Since $\bu \in \cS_{\cF}^{(*,*)}$, 
any prefix of $\bu$ of size $h \times |\bu|_c$ and
any prefix of $\bu$ of size $|\bu|_r \times w$ 
are generated by subgraphs $G_c(\cS)$ and $G_r(\cS)$, respectively.
Therefore, we have
$h < |\bu|_r$ and $w < |\bu|_c$ (since otherwise, $G(\cS)$ can generate $\bu$), 
and there exists a prefix $\bv$ of $\bu$ such that
\begin{align}
&G(\cS)\ \text{does not generate}\ \bv,\label{eq:pGA1}\\
&G(\cS)\ \text{generates}\ \pi_r(\bv),\label{eq:pGA2}\\
&G(\cS)\ \text{generates}\ \pi_c(\bv).\label{eq:pGA3}
\end{align}
Observe that (\ref{eq:pGA2}) and (\ref{eq:pGA3}) 
imply $|\bv|_r>h$ and $|\bv|_c>w$, and therefore, $\bv \in \cS_{\cF}^{(*,*)}$ holds. 

Let $(i,j)$ be the right-bottom coordinate of $\bv$ in $\bu$, so the right-bottom coordinate of $\pi_r(\bv)$ and $\pi_c(\bv)$ are given by $(i\!-\!1, j)$ and $(i, j\!-\!1)$, respectively.
From (\ref{eq:pGA2}),  the suffix of $\pi_r(\bv)$ of size $h \times |\bv|_c$ is an element of $\cS_{\cF}^{(h, *)}$, 
so path $s(i\!-\!1, j\!-\!1) \rar s(i\!-\!1, j)$ exists.
Similarly, 
from (\ref{eq:pGA3}),  the suffix of $\pi_c(\bv)$ of size $|\bv|_r \times w$ is an element of $\cS_{\cF}^{(*, w)}$, so path $s(i\!-\!1, j\!-\!1) \bar s(i, j\!-\!1)$ exists.
Moreover, since $\bv \in \cS_{\cF}^{(*, *)}$,
the suffix of $\bv$ of size $h \times |\bv|_c$ is an element of $\cS_{\cF}^{(h, *)}$
so path $s(i, j\!-\!1) \rar s(i, j)$ exists.
Similarly,
the suffix $\by$ of $\bv$ of size $|\bv|_r \times w$ is an element of $\cS_{\cF}^{(*, w)}$
so path $s(i\!-\!1, j) \bar s(i, j)$ exists.

Therefore, there exist two paths $s(i\!-\!1, j\!-\!1) \rar s(i\!-\!1, j) \bar s(i, j)$
and $s(i\!-\!1, j\!-\!1) \bar s(i, j\!-\!1) \rar s(i, j)$ shown in Fig.~\ref{fig:condT}.
Hence, $G(\cS)\ \text{generates}\ \bv$, which contradicts 
the assumption (\ref{eq:pGA1}) as desired. 
\end{proof}

We note here a very important remark that the arguments we have done so far 
hold even though we generate $G_r(\cS)$ or $G_r(\cS)$ by allowing their vertex set $\cA_{\cF}$ to be a multi-set; that is, 
two or more same blocks can be used as distinct vertices in $G_r(\cS)$ or $G_r(\cS)$. 
That is a key point to generate blocks from $G(\cS)$ which will be discussed in the next section.

\section{How to Generate Blocks from $G(\cS)$\label{howto}}
\subsection{Preliminaries for $G(\cS)$}

From Theorem~\ref{th:GA},  for $\bb \in \cS_{\cF}^{(*, *)}$, 
a subblock $\by$ of size $h \times |\bb|_c$ is generated by a subgraph $G_c(\cS)$. 
So \INS{its $h\times w$ prefix $\bv$} is the head block of the path 
$\tau$ in the $G_c(\cS)$ generating $\by$. 
Since there are $||\cA_{\cF}||$ vertices in $G_c(\cS)$, 
we can consider $||\cA_{\cF}||$ different subgraphs 
$G_c(\cS)$'s with respect to the head block $\bv$, and define the $g(\bv)$\emph{-class of }$G_c(\cS)$ to be the subgraph $G_c(\cS)$ each path in which has $\bv$ as its head block. 

Now suppose that for the block 
$\bv$ above, its right-bottom coordinate in $\bb$ is $(i-1, w)~(i > h)$,
\REP{$\bv$ is the head block of $\by$ whose right-bottom coordinate in $\bb$ is $(i-1, w)~(i > h)$,}
and hence, the right-bottom coordinate of $\by$ in $\bb$ is $(i-1, |\bb|_c)$. 
If 
\REP{$\bv'$ the $h \times |\bb|_c$ subblock}
{$\bv'$ is the $h \times w$ subblock} of $\bb$
whose right-bottom coordinate is  $(i, w)$,  
then $\bv'$ is the head block of 
the path $\tau'$ generating $h \times |\bb|_c$ subblock $\by'$ of $\bb$ whose right-bottom coordinate  is $(i, |\bb|_c)$. In other words, the 
$g(\bv')$-class of $G_c(\cS)$ generates $\by'$. 
In such a case, we can always find blue edges from the $g(\bv)$-class of $G_c(\cS)$ to
the $g(\bv')$-class of $G_c(\cS)$ satisfying condition in Fig.~\ref{fig:condT}, 
by considering $s(i-1,j-1), s(i-1,j)$ and $s(i,j-1), s(i,j)$ to be vertices in $g(\bv)$-class of $G_c(\cS)$ and 
vertices in $g(\bv')$-class of $G_c(\cS)$, respectively (except for the case that some vertices have no attached blue edges). 
We imply such a connection by drawing a blue edge from  
the $g(\bv)$-class of $G_c(\cS)$ to the $g(\bv')$-class of $G_c(\cS)$.

Fig.~\ref{fig:exTorusR} illustrates the $g(\bv)$-classes of $G_c(\cS)$
and the connections between them in $G(\cS)$ 
when $\cA_{\cF}$ in Fig.~\ref{fig:A} is given. 
The subgraph $G_c(\cS)$ in a red circle having the number $k(=g(\bv))$, $1 \le k \le 7=||\cA_{\cF}||$, in black circle represents
the $k$-class of $G_c(\cS)$. 
Note that it is of course possible to follow the similar argument for $G_r(\cS)$ and define $||\cA_{\cF}||$ different
$g(\bu)$-class of $G_r(\cS)$'s, by considering a subblock $\bx$ of size $|\bb|_r \times w$  
and the prefix $\bu$ of $\bx$ of height $h$.

	\begin{figure}[htb]
		\begin{center}
			\includegraphics[width=0.6\linewidth]{\figpath/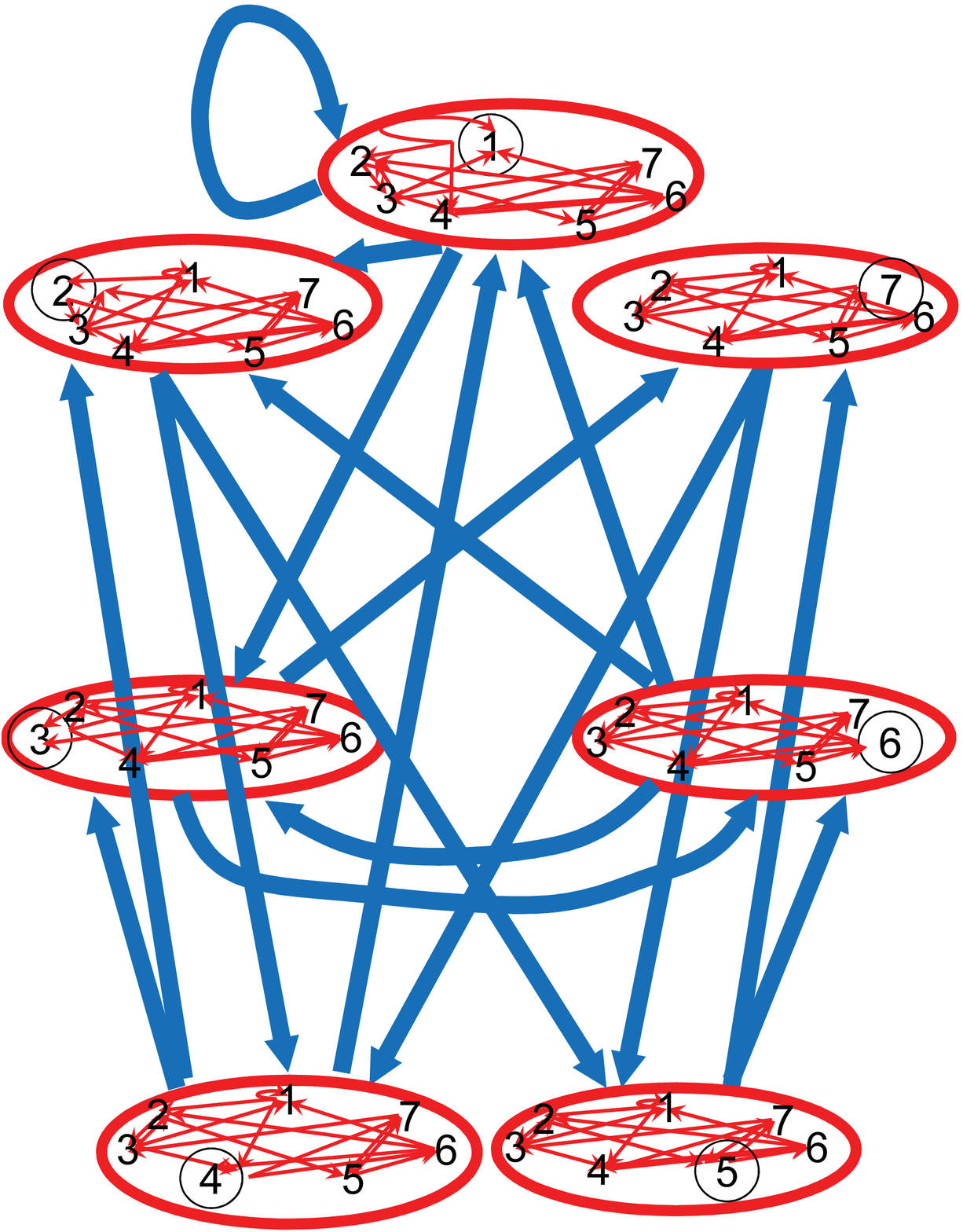}
			\caption{The $g(\bs)$-classes of $G_c(\cS)$ and their connections in $G(\cS)$}
			\label{fig:exTorusR}
		\end{center}
	\end{figure}

\subsection{Process to Generate Blocks in $\cS_{\cF}^{(*,*)}$}
We are now in a position to show how to generate a block $\bb\in \cS_{\cF}^{(*,*)}$ by means of the $g(\bv)$-classes of $G_c(\cS)$
and their connections in $G(\cS)$. 
We explain, as an example, 
a case of generating a $3\times 5$ block $\bb$
using Fig.~\ref{fig:exTorusR} to avoid ambiguity.  
Fig.~\ref{fig:trans} shows a process to generate $\bb$, where
newly generated blocks and identifiers are written by bold numbers in
each step.

In the first step, pick an arbitrary vertex $\bv$ in $\cA_{\cF}$.
Say suppose that $\bv$ is the first block in Fig.~\ref{fig:A} with identifier $1$.
The vertex of identifier $1(=s(2, 2))$ is the head block for both $G_r(\cS)$ and $G_c(\cS)$.
Therefore, 
the $1$-class of $G_r(\cS)$ generates a prefix of $\bb$ of width $w$, and
the $1$-class of $G_c(\cS)$ generates a prefix of $\bb$ of height $h$.

In the second and third steps, $s(2, 3) = 2$ and $s(3, 2) = 3$ are generated by the $1$-class of $G_c(\cS)$
and $G_r(\cS)$, respectively. 
\INSa{Since $s(2(=h), 3)$ and $s(3, 2(=w)) = 3$, the second and the third step
satisfy Case 1 in Remark~\ref{case:rem}.}
Moreover, a vertex of identifier $3 (=s(3, 2))$ is the head block for $G_c(\cS)$, that is the $3$-class of $G_c(\cS)$,
which generates a subblock of the second and third rows.
The fourth, seventh, and eighth steps satisfy paths among four adjacent vertices shown in Fig.~\ref{fig:condT}
since Case 2 in Remark~\ref{case:rem} is satisfied in these steps.
For example, in the forth step, $(1, 2, 3, 4)$ \INSa{at the fifth rectangle in the second} row in Fig.~\ref{fig:exCond} is utilized
to determine $s(3, 3)(=4)$
In the seventh and eighth steps, $(2, 5, 4, 7)$ and $(5, 6, 7, 3)$ are utilized to determine $s(3, 4)(=7)$ and
$s(3, 5)(=3)$, respectively.
The fifth and sixth steps satisfy Case 1 \INSa{in Remark~\ref{case:rem}} so that they are implemented in the $1$-class of $G_c(\cS)$.

Strong points of our method are 
\begin{enumerate}
\item  \INSa{we can freely select any vertex in $\cA_{\cF}$ in the first step;}
\item  we can freely select a step (so a process is not unique); and 
\item  it is possible to make the size of the desired block $\bb$ bigger in the middle of a process
\end{enumerate}
as long as each step satisfies
Case 1 or Case 2 in Remark~\ref{case:rem}.  Thus, we can say that our method is easy to apply and flexible.

	\begin{figure}[htb]
		\begin{center}
			\includegraphics[clip, width=0.76\linewidth]{\figpath/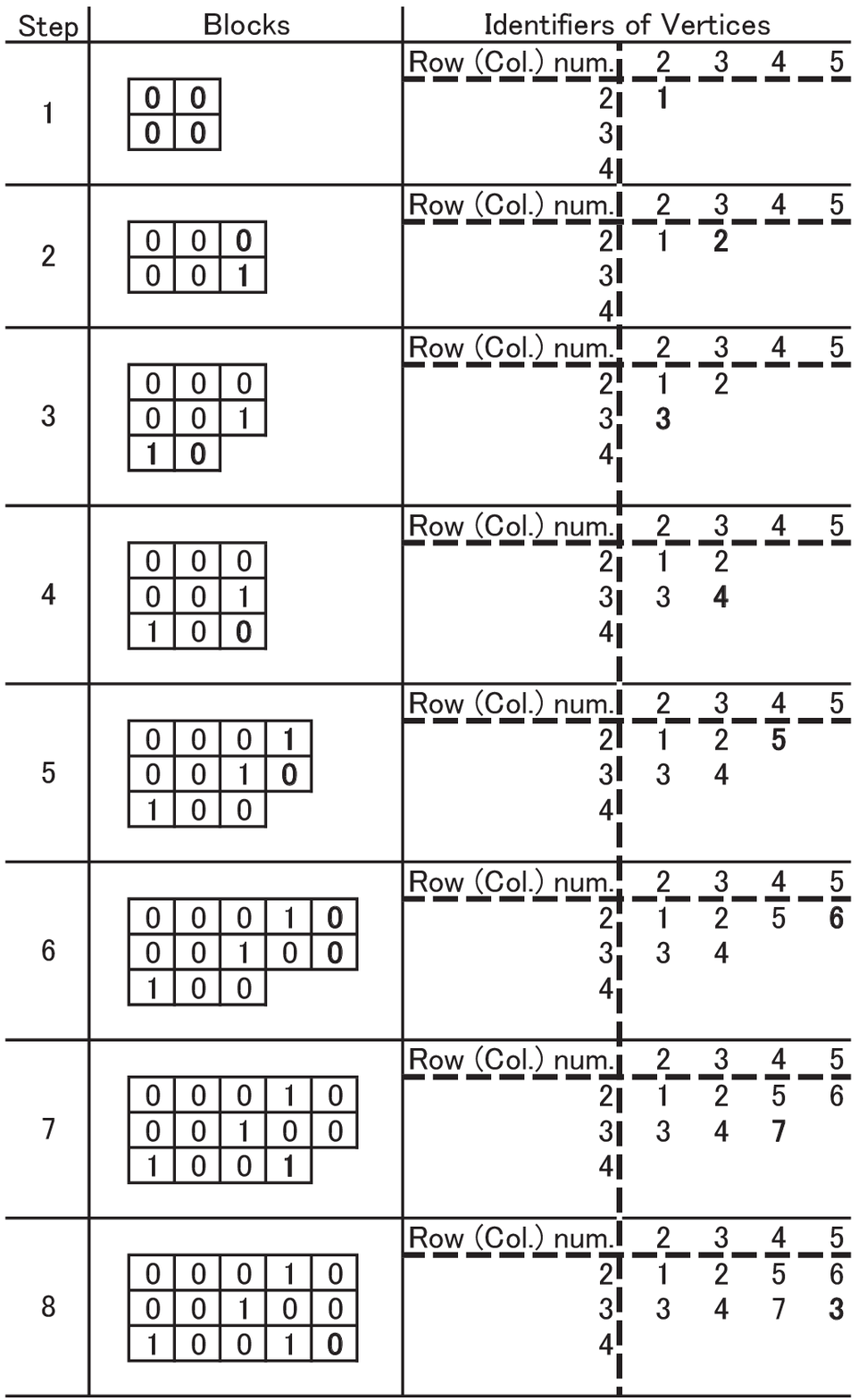}
			\caption{A process to generate a block}
			\label{fig:trans}
		\end{center}
	\end{figure}

\section{Conclusion \label{conclusion}}
In this paper, 
we proposed a graph representation for the 2D-FTCS $\cS=\cS_{\cF}$ (or $\cS_{\cF}^{(*,*)}$
to be more precise)
characterized by a finite set of forbidden set $\cF$. 
More precisely, we proposed a labelled directed graph $G(\cS)$ obtained from 
the row-wise and the column-wise presentations, and then proved that the $G(\cS)$ generates a block if and only if
the block is an element of $\cS_{\cF}^{(*,*)}$. 
We further explained how to indeed generate an arbitrary block in $\cS_{\cF}^{(*,*)}$
from $G(\cS)$. 
\INS{Moreover, the proposed graph representation can be easily extended to three or higher dimensional forbidden set by 
adding a labelled directed graph such as $G_r(\cS)$ for an adding axis.}
As a future work, we hope to apply the
the proposed graph for further analysis in two-dimensional case, such as a 2D antidictionary coding which utilizes
a subset of 2D antidictionary~\cite{OM15Ho, OM14Me2} and the computation of  the capacities of
2D-FTCSs.

\bibliographystyle{ieeetr}
\bibliography{isit16_2DG_v142_arXiv}
\end{document}